\documentclass[11pt]{amsart}
\usepackage{amsmath}
\usepackage{amsthm}
\usepackage{latexsym}
\usepackage{amssymb}
\usepackage{xspace}
\usepackage{amscd}

\newcommand{\bF}{\overline{\mathbb F}}
\newcommand{\PAut}{{\rm PAut}}

\newcommand{\inv}{^{-1}}
\newcommand{\F}{\mathbb F}
\newcommand{\N}{{\mathbb N}}
\newcommand{\Z}{{\mathbb Z}}

\newcommand{\CC}{\mathcal{C}}

\newcommand{\GL}{{\rm GL}}
\newcommand{\PGL}{{\rm PGL}}

\newcommand{\matriz}[1]{\begin{array} #1 \end{array}}
\newcommand{\pmatriz}[1]{\left(\begin{array} #1 \end{array}\right)}
\newcommand{\GEN}[1]{\langle #1 \rangle}

\title{An intrinsical description of group codes}

\author[J.J. Bernal]{Jos\'{e} Joaqu\'{\i}n Bernal}
\author[\'{A}. del R\'{\i}o]{\'{A}ngel del R\'{\i}o}
\author[J.J. Sim\'{o}n]{Juan Jacobo Sim\'{o}n}

\newtheorem{theorem}{Theorem}[section]
\newtheorem{lemma}[theorem]{Lemma}
\newtheorem{proposition}[theorem]{Proposition}

\newtheorem{definition}[theorem]{Definition}
\newtheorem{corollary}[theorem]{Corollary}

\newtheorem{remarks}[theorem]{Remarks}

\newtheorem{example}[theorem]{Example}

\theoremstyle{remark}
\theoremstyle{remark}

\begin{document}
\maketitle

\begin{abstract}
A (left) group code of length $n$ is a linear code which is the image of a (left) ideal
of a group algebra via an isomorphism $\F G \rightarrow \F^n$ which maps $G$ to the
standard basis of $\F^n$. Many classical linear codes have been shown to be group codes.
In this paper we obtain a criterion to decide when a linear code is a group code in terms
of its intrinsical properties in the ambient space $\F^n$, which does not assume an ``a
priori'' group algebra structure on $\F^n$. As an application we provide a family of
groups (including metacyclic groups) for which every two-sided group code is an abelian
group code.

It is well known that Reed-Solomon codes are cyclic and its parity check extensions are
elementary abelian group codes. These two classes of codes are included in the class of
Cauchy codes. Using our criterion we classify the Cauchy codes of some lengths which are
left group codes and the possible group code structures on these codes.
\end{abstract}

\bigskip

In this paper $\F=\F_q$ denotes the field with $q$ elements where $q$ is a power of a
prime $p$. We consider $\F$ as the alphabet of linear codes of length $n$ and so the
ambient space is $\F^n$, the $n$-dimensional vector space. The standard basis of $\F^n$
is denoted by $E=\{e_1,\dots,e_n\}$. For any group $G$, we denote by $\F G$ the group
algebra over $G$ with coefficients in $\F$.

Recall that a linear code $C\subseteq \F^n$ is said to be cyclic if and only if $C$ is
closed under cyclic permutations, that is, $(x_1,\dots,x_n)\in C$ implies
$(x_2,\dots,x_n,x_1)\in C$. For $C_n=\GEN{g}$, the cyclic group of order $n$, the
bijection $\pi:E\rightarrow C_n$ given by $\pi(e_i)=g^{i-1}$ extends to an isomorphism of
vector spaces $\phi:\F^n\rightarrow \F C_n$ and the cyclic codes in $\F^n$ are the
subsets $C$ of $\F^n$ such that $\phi(C)$ is an ideal of $\F C_n$.

If $G$ is a group of order $n$ and $C\subseteq \F^n$ is a linear code
then we say that $C$ is a {\em left $G$-code} (respectively, a {\em right $G$-code}; a
{\em $G$-code}) if there is a bijection $\phi:E\rightarrow G$ such that the linear
extension of $\phi$ to an isomorphism $\phi:\F^n\rightarrow \F G$ maps $C$ to a left
ideal (respectively, a right ideal; a two-sided ideal) of $\F G$. A {\em left group code}
(respectively, {\em group code}) is a linear code which is a left $G$-code (respectively,
a $G$-code) for some group $G$. A (left) cyclic group code (respectively, abelian group
code, solvable group code, etc.) is a linear code which is (left) $G$-code for some
cyclic group (respectively, abelian group, solvable group, etc.). In general, if
$\mathcal{G}$ is a class of groups then we say that a linear code $C$ is a (left)
$\mathcal{G}$ group code if and only if $C$ is a (left) $G$-code for some $G$ in
$\mathcal{G}$.

Note that the cyclic codes of length $n$ are $C_n$-codes. However, not every $C_n$-code
is a cyclic code. For example, the linear code $\{(a,a,b,b):a,b\in \F\}$ is not a cyclic
code but it is a $C_4$-code via the map $\phi:\{e_1,\dots,e_4\}\rightarrow C_4$ given by
$\phi(e_1)=1$, $\phi(e_2)=g^2$, $\phi(e_3)=g$ and $\phi(e_4)=g^3$. If one fixes a
bijection $\phi:E\rightarrow G$ to induce an isomorphism $\phi:\F^n\rightarrow \F G$ then
the (left) $G$-codes are precisely those codes of $\F^n$ which are permutation equivalent
to codes of the form $\phi\inv(I)$, for $I$ running on the (left) ideals of $\F G$. In
particular, the cyclic group codes are the codes which are permutation equivalent to
cyclic codes.

The aim of this paper is to give a criterion to decide if a linear code is a group code in
terms of its intrinsical properties in the ambient space, which has not an ``a priori''
group algebra structure. More precisely we prove that a subspace $C$ of $\F^n$ is a left
group code if and only if $\PAut(C)$, the group of permutation automorphisms of $C$,
contains a regular subgroup. Moreover, $C$ is a group code if and only if $\PAut(C)$
contains a regular subgroup $H$ such that $C_{S_n}(H)\subseteq \PAut(C)$. Using this we
prove that if $G=AB$, for $A$ and $B$ two abelian subgroups of $G$ then every $G$-code is
also an abelian group code. This extends a result of Sabin and Lomonaco \cite{SL}. We
present some examples showing that this cannot be extended to left group codes. Finally, we
study the class of Cauchy codes. This is a class of MDS codes introduced in \cite{D}, (see
also \cite{H}) which includes the class of generalized Reed-Solomon codes. We describe the
possible left group code structures of some families of Cauchy codes.

\section{Characterizing group codes}

In this section we give necessary and sufficient conditions for a subspace $C$ of $\F^n$
to be a (left) group code in terms of the group of permutation automorphisms of $C$. For this result we do not need to assume that $\F$ is finite.

Let $S_n$ denote the group of permutations on $n$ symbols, that is the group of
bijections of $\N_n=\{1,\dots,n\}$ onto itself. Recall that a subgroup $G$ of $S_n$ is
{\em regular} if and only if it is transitive and has order $n$ (equivalently, $|G|=n$
and $\sigma(x)\ne x$ for every $1\ne \sigma \in G$ and $x\in \N_n$). We consider $S_n$
acting by linear transformations on $\F^n$ via the rule:
    $$\sigma(e_i) = e_{\sigma(i)} \quad (\sigma\in S_n, i\in \N_n).$$

Two linear codes $C$ and $C'$ are permutation equivalent if $\sigma\left(C\right)=C'$ for some $\sigma \in S_n$. Following \cite{H}, we denote the group of permutation automorphisms of a linear code $C$
by
    $$\PAut(C)=\{\sigma\in S_n : \sigma(C)=C\}.$$

For a group $G$, let $Z(G)$ denote the center of $G$ and let $C_G(H)$ denote the
centralizer of $H$ in $G$, for each subgroup $H$ of $G$.

We will need the following lemma which is possibly well known.

\begin{lemma}\label{regular}
Let $H$ be a regular subgroup of $S_n$ and fix an element $i_0\in \N_n$. Let
$\psi:H\rightarrow \N_n$ be the bijection given by $\psi(h)=h(i_0)$. Then there is an
anti-isomorphism $\sigma:H\rightarrow C_{S_n}(H)$, mapping $h\in H$ to $\sigma_h$, where
    $$\sigma_h(i)=\psi\inv(i)\left( h(i_0) \right) \quad (i\in \N_n).$$

Moreover $\sigma_h=h$ for every $h\in Z(H)$ and so $Z(H)=Z(C_{S_n}(H))$.
\end{lemma}

\begin{proof}
It is clear that $\psi$ is a bijection, since $H$ is regular. Furthermore
$\psi\inv(i)(i_0) = i$, for every $i\in \N_n$ and $\psi\inv\left( h(i_0) \right)=h$, for
$h\in H$. Therefore, if $h,k\in H$ then
    $$\sigma_h k(i) = \sigma_h\left(k\left( \psi\inv(i)(i_0) \right) \right) =
    \psi\inv\left(\left[ k  \psi\inv(i)\right](i_0) \right) \left( h(i_0) \right)
    = \left[k\psi\inv(i)\right]\left( h(i_0) \right) = k \sigma_h(i).$$
This shows that $\sigma$ maps $H$ into $C_{S_n}(H)$.

Now we prove that $\sigma$ is an anti-homomorphism. If $h,k\in H$ and $i\in \N_n$ then
    $$\sigma_{hk}(i) = \psi\inv(i)hk(i_0) = \psi\inv\left(\psi\inv(i)h(i_0)\right)k(i_0)
    = \sigma_k\left(\psi\inv(i)h(i_0)\right) = \sigma_k(\sigma_h(i)).$$

If $\sigma_h=1$ then $\psi\inv(i)(i_0)=i=\sigma_h(i)= \psi\inv(i)h(i_0)$, for every $i\in
\N_n$. Thus $\psi\inv(i) = \psi\inv(i)h$, since $H$ is regular, and hence $h=1$. This
shows that $\sigma$ is injective.

Finally we prove that $\sigma(H)=C_{S_n}(H)$. Indeed, for each $x\in C_{S_n}(H)$ let
$h=\psi\inv(x(i_0))\in H$. Then $h(i_0) = \psi\inv(x(i_0))(i_0) = x(i_0) =
x\left(\left(\psi\inv(i)\right)\inv(i)\right) = \left(\psi\inv(i)\right)\inv x(i)$, for
every $i\in \N_n$. Hence $x(i) = \psi\inv(i) h(i_0) = \sigma_h(i)$ for every $i \in \N$ and
we have  $x=\sigma_h$.

The last calculation also shows that if $x\in Z(H)=H\cap C_{S_n}(H)$ then $h=x=\sigma_h$.
\end{proof}

We are ready to present our intrinsical characterization of (left) group codes.

\begin{theorem}\label{characterization}
Let $C$ be a linear code of length $n$ over a field $\F$ and $G$ a finite group of order
$n$.
\begin{enumerate}
\item $C$ is a left $G$-code if and only if $G$ is isomorphic to a transitive subgroup of $S_n$ contained in
$\PAut(C)$.
\item $C$ is a $G$-code if and only if $G$ is isomorphic to a transitive subgroup $H$ of $S_n$ such that $H\cup
C_{S_n}(H)\subseteq \PAut(C)$.
\end{enumerate}
\end{theorem}

\begin{proof}
Given a bijection $\phi:E\rightarrow G$, we also use $\phi$ to denote its linear
extension $\F^n\rightarrow \F G$, and we define $f=f_{\phi}:G\rightarrow S_n$ by setting
$e_{f(g)(i)}=\phi\inv(g\phi(e_i))$, for $g\in G$ and $i\in \N_n$. Then $f$ is a group
homomorphism.
Furthermore, $f(g)(i)=i$ if and only if $\phi\inv(g\phi(e_i))=e_i$ if and only if
$g\phi(e_i)=\phi(e_i)$ if and only if $g=1$. This shows that $f$ is injective and $H=f(G)$
is a regular subgroup of $S_n$. Let $e_{i_0}=\phi\inv(1)$. By Lemma~\ref{regular}, there is
an anti-isomorphism $\sigma:H\rightarrow C_{S_n}(H)$, defined by
$\sigma_h(i)=\psi\inv(i)h(i_0)$, where $\psi:H\rightarrow \N_n$ is the bijection given by
$\psi(h)=h(i_0)$. We will use several times the equality $\psi\inv(i)(i_0)=i$, for $i\in
\N_n$.

For every $i\in \N_n$ one has
    $$e_i = e_{\psi\inv(i)(i_0)} = e_{f (f\inv \psi\inv(i))(i_0)} =
    \phi\inv\left( f\inv \left( \psi\inv(i)  \right) \cdot \phi(e_{i_0}) \right) =
    \phi\inv f\inv \psi\inv(i).$$
Using this and the fact that $f:G\rightarrow H$ is an isomorphism, one has the following
equalities for every $h\in H$:
    \begin{eqnarray*}
    \sigma_h(e_i) &=& e_{\sigma_h(i)} = e_{(\psi\inv(i)h)(i_0)}
        = \phi\inv f\inv \psi\inv\left( (\psi\inv(i)h)(i_0) \right) \\
    &=& \phi\inv f\inv (\psi\inv(i)h) = \phi\inv \left(f\inv(\psi\inv(i)) \cdot f\inv(h) \right) =
        \phi\inv \left(\phi(e_i) \cdot f\inv(h) \right).
    \end{eqnarray*}
This implies
    \begin{equation}\label{Clave}
    \sigma_h(x) = \phi\inv(\phi(x) \cdot f\inv(h))
    \end{equation}
for every $x\in \F^n$ and $h\in H$.

Assume now that $C$ is a left $G$-code and let $\phi:E\rightarrow G$ be a bijection such
that $\phi(C)$ is a left ideal of $\F G$ and $f=f_{\phi}$. By the previous paragraph
$H=f(G)$ is a regular subgroup of $S_n$. Furthermore, by the definition of $f$, if $h\in
H$ and $\alpha=\sum \alpha_i e_i \in \F^n$, then
    $h(\alpha) = \sum \alpha_i e_{h(i)} =
    \sum \alpha_i \phi\inv \left( f\inv(h) \cdot \phi(e_i) \right)
    = \phi\inv \left( f\inv(h) \cdot \phi(\alpha) \right)$.
Hence $h(C) = \phi\inv\left(f\inv(h) \cdot \phi(C)\right) = C$, that is $H\subseteq
\PAut(C)$. If, moreover, $\phi(C)$ is a two-sided ideal of $\F G$ and $x\in C_{S_n}(H)$
then, by Lemma~\ref{regular}, there is $h\in H$ such that $x=\sigma_h$. Then $x(C) =
\sigma_h(C) = \phi\inv\left(\phi(C)\cdot f\inv(h)\right) = C$, by (\ref{Clave}). This gives
the sufficiency part of (1) and (2).

Conversely, assume that $G$ is isomorphic to a regular subgroup $H$ of $S_n$ contained in
$\PAut(C)$. One may assume without loss of generality that $G=H$. Let $\phi:E\rightarrow
G$ be the bijection given by $\phi\inv(g)=e_{g(1)}$. Then $h\phi(e_{g(1)}) = hg =
\phi(e_{(hg)(1)}) = \phi(h(e_{g(1)}))$, for every $h,g\in G$. Therefore
$h\phi(e_i)=\phi(h(e_i))$, for every $i\in \N_n$, since $G$ is transitive. Using this we
have $h\phi(C) = \phi(h(C))=\phi(C)$, because $G\subseteq \PAut(C)$. Thus $\phi(C)$ is a
left ideal of $\F G$.

Assume, additionally that $C_{S_n}(G)\subseteq \PAut(C)$ and let $f=f_{\phi}$. If $g\in G$
then $f(g)(e_i) = e_{f(g)(i)} = \phi\inv(g\phi(e_i)) = e_{(g\phi(e_i))(1)} =
g(e_{\phi(e_i)(1)}) = g(e_i)$. Thus $f(g)=g$ for every $g\in G$. By (\ref{Clave}), $\phi(C)
g = \phi(C)f\inv(g)=\phi(\sigma_g(C)) = \phi(C)$, for every $g\in G$, and thus $\phi(C)$ is
a two-sided ideal of $\F G$.
\end{proof}

\begin{corollary}
Let $C$ be a linear code of length $n$ over a field and let $\mathcal{G}$ be a class of
groups.
\begin{enumerate}
\item $C$ is a left $\mathcal{G}$ group code if and only if $\PAut(C)$ contains a regular subgroup of $S_n$ in
$\mathcal{G}$.
\item $C$ is a $\mathcal{G}$ group code if and only if $\PAut(C)$ contains a regular subgroup $H$ of $S_n$ in
$\mathcal{G}$ such that $C_{S_n}(H)\subseteq \PAut(C)$.
\end{enumerate}
\end{corollary}

\section{One-dimensional group codes}

Even if one-dimensional codes are not relevant for applications, it is educating to see
how one can use Theorem~\ref{characterization} to decide which one-dimensional codes are
group codes.
%

\begin{proposition}\label{Dim1}
Let $\F$ be a finite field, $n\ge 1$ and $0\ne v\in \F^n$. Then $\F v$, the
one-dimensional subspace spanned by $v$, is a left group code if and only if the entries
of $v$ are of the form $u,\xi u,\xi^2 u,\dots,\xi^{h-1} u$, each of them appearing
$s=n/h$ times in $v$, for some $u\in \F^*$ and $\xi$ an $h$-th primitive root of unity.

Assume that $v$ satisfies the above conditions and let $G$ be a group. Then $\F v$ is a
left $G$-code if and only if $G$ has a normal subgroup $N$ of order $s$ such that $G/N$
is cyclic (of order $h$).
\end{proposition}

\begin{proof}
Write $v=(v_1,\dots,v_n)$. We start describing $\PAut(\F v)$. Let $\sigma\in S_n$. Then
$\sigma\in \PAut(\F v)$ if and only if $\sigma(v)=\lambda_{\sigma} v$, for some
(necessarily unique) $\lambda_{\sigma}\in \F^*$. Moreover, $\sigma\mapsto
\lambda_{\sigma}$ defines a group homomorphism $\lambda:\PAut(\F v)\rightarrow \F^*$. Let
$\xi$ be a generator of the image of $\lambda$ and let $K=Ker \lambda$. Then
$\sigma(v)=\xi v$ and hence, if $u\in \F$ appears exactly $s$ times as an entry of $v$
then $\xi u$ also appears exactly $s$ times as an entry of $v$. Therefore, after
permuting the entries of $v$, one may assume that $v$ is of the form
    $$
    \matriz{{rcl}
    v&=&(
    u_1,\dots,u_1,\xi u_1,\dots,\xi u_1,\xi^2 u_1,\dots,\xi^2 u_1,
    \xi^{h-1}u_1,\dots,\xi^{h-1}u_1,\\
    & & \hspace{0.2cm} u_2,\dots,u_2,\xi u_2,\dots,\xi u_2,\xi^2 u_2,\dots,\xi^2 u_2,
    \xi^{h-1}u_2,\dots,\xi^{h-1}u_2,\\
    & & \hspace{0.2cm} \dots,\\
    & & \hspace{0.2cm} u_k,\dots,u_k,\xi u_k,\dots,\xi u_k,\xi^2 u_k\dots,\xi^2 u_k,,
    \xi^{h-1}u_k,\dots,\xi^{h-1}u_k).}$$
for $u_1,\dots,u_k\in \F$ such that $u_x\ne \xi^z u_y$ for every $1\le x < y\le k$ and
$z\in \Z$, and each $\xi^z u_i$ appears $s_i$ times as an entry of $v$. Then
$K=(S_{s_1}\times \dots \times S_{s_1})\times (S_{s_2}\times \dots \times S_{s_2}) \times
\dots \times (S_{s_k}\times \dots \times S_{s_k})$, where the $i$-th block has $h$ copies
of $S_{s_i}$ and the $j$-th $S_{s_i}$ is identified with the subgroup of $S_n$ formed by
the permutations which fix the coordinates $k$ for which $v_k\ne \xi^j u_i$. Consider the
permutation $\sigma\in S_n$ given by
    $$\sigma(k)=\left\{ \matriz{{ll}
    k+s_i, & \text{if } h(s_1+\dots+s_{i-1})<k\le h(s_1+\dots+s_{i-1})+(h-1)s_i; \\
    k-(h-1)s_i,& \text{if } h(s_1+\dots+s_{i-1})+(h-1)s_i<k\le h(s_1+\dots+s_i).}
    \right.$$
Clearly, $\sigma$ has order $h$ and $\lambda_{\sigma}=\xi$. Thus $\PAut(\F v)=K\rtimes
\GEN{\sigma}$. This provides the description of $\PAut(\F v)$.

Observe that, with the notation of the previous paragraph, $\PAut(\F v)$ is transitive if
and only if $k=1$. In that case (up to a permutation of coordinates)
    \begin{equation}\label{v}
    v=(u,\dots,u,\xi u,\dots,\xi u,\xi^2 u,\dots,\xi^2 u, \xi^{h-1}u,\dots,\xi^{h-1}u)
    \end{equation}
for some $u\in \F^*$ and some $h$-th primitive root of unity $\xi$ in $\F^*$ and each
$\xi^j u$ is repeated $s=n/h$ times as a coordinate of $v$.

If $\F v$ is a left $G$-code, for $G$ a group (of order $n$) then, by
Theorem~\ref{characterization}, we may assume that $G$ is a transitive subgroup of $S_n$
contained in $\PAut(\F v)$. Therefore $v$ is as in (\ref{v}) and $G$ contains an element
$\tau$ with $\tau(1)=h+1$. This implies that $\tau=k\sigma$ for some $k\in K$. If
$N=K\cap G$ then $G=\GEN{N,\tau}$ and $\tau$ has order $h$ modulo $N$. Thus $N$ is a
normal subgroup of index $h$ in $G$ and $G/N$ is cyclic. This proves one implication for the two statements of the proposition.

Conversely, assume that $v$ is as in (\ref{v}) and let $G$ be a group of order $n$ with a
normal subgroup $N$ of index $h$ in $G$ such that  $G/N$ is cyclic. Let $g$ be an element
of $G$ of order $h$ modulo $N$. Fix a bijection $\varphi:\N_s \rightarrow N$, which we
extend to obtain a bijection $\varphi:\N_n\rightarrow G$ by setting
$\varphi(js+r)=\varphi(r)g^j$, for $1\le r \le s$, and $0\le j < h$. Then $\varphi$
induces an isomorphism of vectors spaces $\varphi:\F^n\rightarrow \F G$, such that
$\varphi(v)=u\left(\sum_{n\in N} n \right)\left(\sum_{i=0}^{h-1} \xi^i g^i\right)$. Then
$g\varphi(v)=\xi\inv \varphi(v)$ and $n\varphi(v)=\varphi(v)$ for every $n\in N$. Thus
$\F\varphi(v)$ is a left ideal of $\F G$ and hence $\F v$ is a left $G$-code.
\end{proof}

\begin{corollary}
If a one-dimensional vector space is a left group code then it is a cyclic group code.
\end{corollary}

\begin{proof}
Assume that $\F v$ is a one-dimensional group code of length $n$. By
Proposition~\ref{Dim1}, $v$ has $h$ different coordinates each of them repeated $s=n/h$
times. Then $\F v$ is a cyclic group code because the cyclic group of order $n$ has a
subgroup of order $s$.
\end{proof}

\section{Group codes versus abelian group codes}

In this section we present a class of groups $\mathcal{G}$, containing properly the class
of abelian groups, such that every $\mathcal{G}$ group code is also an abelian group
code. We also show that for every non-abelian finite group $G$ there is a finite field
$\F$ such that some left $G$-code over $\F$ is not an abelian group code.

\begin{theorem}\label{TwoSidedAbelian}
Let $G$ be a finite group. Assume that $G$ has two abelian subgroups $A$ and $B$ such
that every element of $G$ can be written as $ab$ with $a\in A$ and $b\in B$. Then every
$G$-code is an abelian group code.
\end{theorem}

\begin{proof}
Assume that $C$ is a $G$-code of length $n$. Then, by Theorem~\ref{characterization}, $G$
is isomorphic to a regular subgroup $H$ of $S_n$ such that $H\cup C_{S_n}(H)\subseteq
\PAut(C)$. One may assume without loss of generality that $G=AB=H$.

Let $\sigma:G\longrightarrow C_{S_n}(G)$ be the anti-isomorphism defined in Lemma
\ref{regular}, with respect to some $i_0\in \N_n$. Set $A_1=\sigma(A)$ and $B_1=\sigma(B)$.
Since $\sigma$ is an anti-isomorphism, $A_1$ and $B_1$ are abelian groups and
$C_{S_n}(G)=B_1 A_1$.

Now, we consider the group $K=\GEN{A, B_1}\le \GEN{G,C_{S_n}(G)} \le \PAut(C)$. Then $K$ is
abelian, since $A$ and $B_1$ are abelian subgroups, $A\subseteq G$ and $B_1\subseteq
C_{S_n}(G)$. By Theorem~\ref{characterization}, it is enough to prove that $K$ has order
$n$ and the stabilizer in $K$ of every element of $\N_n$ is trivial. To prove the first it
is sufficient to show that $\left[K:A\right]=\left[G:A\right]$. By Lemma \ref{regular},
$B\cap Z(G)=\sigma(B\cap Z(G))=B_1\cap Z(C_{S_n}(G))$. Then $[B:B\cap Z(G)]=[B_1:B_1\cap
Z(C_{S_n}(G))]$ and $\left[B\cap Z(G):B\cap A\right]=\left[B_1\cap Z(C_{S_n}(G)):B_1\cap
A\right]$. Thus $\left[G:A\right]=\left[AB:A\right]=\left[B:B\cap
A\right]=\left[B_1:B_1\cap A\right]= \left[AB_1:A\right]=\left[K:A\right]$, as desired.

Finally, let $i\in \N_n$ and $k\in K$ with $k(i)=i$. Then $k=a \beta$ for some $a\in A$
and $\beta=\sigma_{b}\in B_1$ with $b\in B$. Moreover, $i=g(i_0)$ for a unique $g\in G$.
Therefore $g(i_0)=i=k(i)=a \beta g(i_0)=a \sigma_{b}(\psi(g)) = a g b(i_0)$, where $\psi$
is the map defined in Lemma~\ref{Clave}. Thus $a g b=g$, because $G$ is regular. Set
$g=a' b'$, for $a'\in A$ and $b'\in B$. Then $a' a b b'=a a' b' b = agb=g=a'b'$ and hence
$a b=1$. Thus $b^{-1}=a\in A\cap B \subseteq Z(G)$. Using the last part of
Lemma~\ref{Clave} once more, we have $\beta=\sigma_{b}=b$ and so $a \beta=1$.
\end{proof}

Recall that a group $G$ is metacyclic if it has a cyclic normal subgroup $N$ such that
$G/N$ is cyclic. Sabin and Lomonaco \cite{SL} have proved that if $C$ is a $G$-code of a
split metacyclic group $G$ (i.e. $G$ is a semidirect product of cyclic groups) then $C$
is an abelian group code. Using the previous result this can be extended to arbitrary
metacyclic group codes.

\begin{corollary}
If $C$ is a metacyclic group code then $C$ is an abelian group code.
\end{corollary}

Theorem~\ref{TwoSidedAbelian} provides a family of groups $\mathcal{G}$ such that every
two-sided $\mathcal{G}$ group code is also an abelian group code. We do not know any
example of a group code which is not an abelian group code. For left group codes the
situation is completely different.

\begin{proposition}\label{Left}
For every non-abelian group $G$ and every prime $p$ not dividing the order of $G$ there is a left
$G$-code over some field of characteristic $p$ which is not an abelian group code.
\end{proposition}

\begin{proof}
Assume that $p$ is a prime not dividing $|G|$. If $c$ is the number of conjugacy classes of $G$ then the number of
ideals of $\F G$ is at most $2^c$, for every field $\F$ of characteristic $p$. Therefore, there is a bound on the
number of $G$-codes which is independent of the field. In particular, the number of abelian group codes of a given
length is bounded by a number not depending on the field used. However, if $G$ is non-abelian, then there is a field
$K$ such that a simple factor of $KG$ is isomorphic to $M_m(K)$ for some $m\ge 2$. Then a simple factor of $\F G$ is
isomorphic to $M_m(\F)$ for every field extension $\F/K$. Since the left ideals of $M_m(\F)$ generated by matrices of
the form
    $$\pmatriz{{ccccc}
    1 & a & 0 & \dots & 0 \\
    0 & 0 & 0 & \dots & 0 \\
    \dots & \dots & \dots & \dots & \dots \\
    0 & 0 & 0 & \dots & 0
    } \quad (a\in \F)$$
are all different, $\F^n$ has at least $|\F|$ different left $G$-codes. Thus, the number
of left $G$-codes for the different fields of characteristic $p$ increases with the order
of the field and so it is not limited.
\end{proof}

\begin{example}{\rm
Since the proof of Proposition~\ref{Left} is not constructive it is worth to present an
example of a left group code which is not an abelian group code. Let $\F=\F_{11}$ be the
field with 11 elements. Every $2$-dimensional cyclic group code of length 6 over $\F$ is
permutation equivalent to one of the following codes:
    $$\matriz{{rcl}
    C_1&=&\{\left(\lambda,\mu,\mu-\lambda,-\lambda,-\mu,\lambda-\mu\right) \mid
    \lambda,\mu\in \F \}, \\
    C_2&=&\{\left(\lambda,\mu,-\mu-\lambda,\lambda,\mu,-\mu-\lambda\right) \mid
    \lambda,\mu\in \F \}, \\
    C_3&=&\{\left(\lambda,\mu,\lambda,\mu,\lambda,\mu\right)\mid \lambda,\mu\in \F \}.}$$
Then one can easily check that the subspace $C$ of $\F^6$ generated by
$u=\left(2,5,-7,2,-7,5\right)$ and $v=\left(4,-3,-1,-4,1,3\right)$ is not an abelian
group code. However $C$ is a left $S_3$-code. Indeed, $A=(1,2,3)(4,5,6)$ and
$B=(1,4)(2,6)(3,5)$ belong to $\PAut(C)$, since $A(u)=5u+4v$, $B(u)=u$, $A(v)=5(v-u)$ and
$B(v)=-v$. Moreover, $\GEN{A,B}$ is a regular subgroup of $S_6$ isomorphic to $S_3$ and
the claim follows from Theorem~\ref{characterization}.
 }\end{example}

\section{Cauchy codes}

In this section we give some necessary and sufficient conditions for some Cauchy codes to
be left group codes and describe the possible left group code structures on such codes.
Our main tool is a description of the group of permutation automorphisms due to D\"{u}r
\cite{D} (see also \cite{H}) which we present in a slightly different form.

We first recall the definition of Cauchy codes. For that we need to introduce some  notations. Let $\bF$ denote
the projective line over $\F$. Each element of $\bF$ is represented either by its homogeneous coordinates, denoted
$[x,y]$, or by an element of $\F\cup \{\infty\}$ in such a way that $x\in \F$ and $[x,1]$ represent the same element in
$\bF$, and $\infty$ and $[1,0]$ also represent the same element in $\bF$. We use the coordinatization map $\varphi:\bF
\rightarrow \F^2$ defined as follows:
    $$\varphi(x) =
    \left\{\matriz{{ll} (x,1), & \text{if } x\in \F; \\ (1,0), & \text{if } x=\infty.}\right.$$

We evaluate a  homogeneous polynomial $F\in \F[X,Y]$ on an element $z\in \bF$ by setting $F(z)=F(\varphi(z))$. For
an element $\alpha=(\alpha_1,\dots,\alpha_n)\in \bF^n$ we denote the set of entries of $\alpha$ by $L_{\alpha}$.

\begin{definition}
 Given $1\le k < n$, $\alpha=(\alpha_1,\dots,\alpha_n)\in
\bF^n$, without repeated entries, and a map $f:L_{\alpha}\rightarrow \F^*$, the $k$-dimensional
{\em Cauchy code} with location vector $\alpha$ and scaling map $f$ is the code
    $$\CC_k(\alpha,f)=\left\{(f(\alpha_1)P(\alpha_1), \dots, f(\alpha_n)P(\alpha_n)):
    P\in \F[X,Y], \text{ homogeneous of degree } k-1\right\}.$$
\end{definition}

Note that our notation is slightly different from that of \cite{H}. The set $L_{\alpha}$
is called the location set of $\CC_k(\alpha,f)$.

Every Cauchy code is MDS and its dual is another Cauchy code (see \cite{H} for details).
Moreover, if $\sigma\in S_n$ then $\sigma(\CC_k(\alpha,f))=\CC(\sigma(\alpha),f)$ and in
particular the class of Cauchy codes is closed under permutations.

Examples of Cauchy codes are the generalized Reed-Solomon codes which correspond to
Cauchy codes with location set included in $\F$. In particular, the $k$-dimensional narrow
sense Reed-Solomon codes are those Cauchy codes of the form $\CC_k(\alpha,1)$ with location
vector $\alpha=(1,\xi,\xi^2,\dots,\xi^{q-2})$, for $\xi$ a primitive element of $\F$.
Furthermore the parity check extension of the narrow sense Reed-Solomon code
$\CC_k(\alpha,1)$ is the Cauchy code $\CC_k(\overline{\alpha},1)$ with
$\overline{\alpha}=(1,\xi,\xi^2,\dots,\xi^{q-2},0)$. (Here $1$ denotes the constant map
$x\mapsto 1$.)

As mentioned at the beginning of this section our aim is to obtain criteria to decide when
a Cauchy code $\CC_k(\alpha,f)$ is a left $G$-code.  Since a linear code and its dual
have the same automorphism groups, they are (left) $G$-codes for the same groups $G$.
Furthermore the dual of a Cauchy code is a Cauchy code too and $\CC_1(\alpha,f)$ is the
one-dimensional code generated by $(f(\alpha_1),\dots,f(\alpha_n))$ (see \cite{H} for
details). This shows that the cases $k=1$ and $k=n-1$  reduce to Proposition~\ref{Dim1}.
Thus, in the remainder of this section, we assume that $2\le k \le n-2$.

Let $\Gamma=\PGL_2(\F)$, the 2-dimensional projective general linear group. We identify
$\Gamma$ with the group of homographies of $\bF$, so that $\GL_2(\F)$ is considered
acting on the left on $\bF$ by
    $$\pmatriz{{cc} a&b\\c&d} [x,y] = [ax+by , cx+dy].$$
For a subset $X$ of $\bF$ let $\Gamma_X=\{T\in \Gamma:T(X)=X\}$. For $a\in \bF$, let
$\Gamma_a=\Gamma_{\{a\}}$. For example, $\Gamma_{\infty}$ is formed by the classes
represented by upper triangular invertible matrices and
    \begin{equation}\label{NH}
    \Gamma_{\infty}\simeq N\rtimes H, \text{ with }
    N=\left\{ U_b=\pmatriz{{cc} 1 & b \\ 0 & 1}:b\in \F \right\}
    \text{ and } H=\left\{\pmatriz{{cc} a & 0 \\ 0 & 1}: a\in \F^*\right\}.
    \end{equation}
Moreover $N$ is an elementary abelian $p$-group of order $q$ and $H$ is cyclic of order
$q-1$.

For every $T\in \GL_2(\F)$, let $\theta_T:\bF\rightarrow \F^*$ be defined by
    $$\theta_T(z)=\left\{\matriz{{ll} \text{second coordinate of } T(\varphi(z)), & \text{ if it is non-zero}; \\
    \text{first coordinate of } T(\varphi(z)), & \text{ otherwise}.}\right.$$

Fix a subset $L$ of $\bF$. If $T\in \GL_2(\F)$ then we denote the restrictions of
$T:\bF\rightarrow \bF$ and $\theta_T:\bF\rightarrow \F^*$ to $L$ as $T_L$ and
$\theta_{T,L}$ respectively. If $f,g:L\rightarrow \F^*$ then let $f\cdot g: L \rightarrow
\F^*$ be given by
    $$(f\cdot g)(z)=f(z)g(z) \quad (z\in L).$$
We say that $f$ and $g$ are equivalent, denoted $f\equiv g$, if $g=\lambda f$ for some
$\lambda\in \F^*$.

The following theorem describes when two Cauchy codes are equal.

\begin{theorem}\label{Equal}\cite[Theorem~2.6]{H}
Let $2\le k \le n-2$ and consider two $k$-dimensional Cauchy codes.
$C=\CC_k(\alpha,f)$ and $C'=\CC_k(\alpha',f')$ of length $n$. Then $C=C'$ if and only if there exists
$T\in \PGL_2(\F)$ such that $\alpha'_i=T(\alpha_i)$ and $f'\circ T_{L_{\alpha}} \equiv
\theta_{T,L_{\alpha}}^{k-1}\cdot f$.

In particular, $C$ and $C'$ are permutation equivalent if and only if there exists $T\in
\PGL_2(\F)$ such that $T(L_{\alpha})=L_{\alpha'}$ and $f'\circ T_{L_{\alpha}} \equiv
\theta_{T,L_{\alpha}}^{k-1}\cdot f$.
\end{theorem}

Note that the equivalence $f'\circ T_{L_{\alpha}} \equiv \theta_{T,L_{\alpha}}^{k-1}\cdot
f$, in Theorem~\ref{Equal} does not depend on the representative of $T$ in $\GL_2(\F)$
used to calculate $\theta_T$, because $\theta_{\lambda T}=\lambda \theta_T$, for every
$\lambda\in \F^*$.

Given $f:L\rightarrow \F^*$ with $|L|=n$ we denote
    $$\Gamma_{k,f} = \{T\in \Gamma_L : f\circ T_L \equiv \theta_{T,L}^{k-1}\cdot
    f\}$$

\begin{proposition}
If $C$ is a $k$-dimensional Cauchy code of length $n$ with $2\le k \le n-2$ and scaling
map $f:L\rightarrow \F^*$ then the map $T\mapsto T_L$ is a group isomorphism
$\Gamma_{k,f}\rightarrow \PAut(C)$.
\end{proposition}

\begin{proof}
Let $C=\CC_k(\alpha,f)$ with $L=L_{\alpha}$. If $\sigma \in S_n$ then
$\sigma(C)=\CC_k(\sigma(\alpha),f)$. Using this and Theorem~\ref{Equal} and identifying
$S_n$ with the permutations of $L$ one has
    \begin{equation}\label{PAut}
    \PAut(C) = \left\{T_L : T\in \Gamma_{L,f}\right\}.
    \end{equation}
The map $T\mapsto T_L$ is injective because $|L|=n>k\ge 2$ and an element of $\PGL_2(K)$
is uniquely determined by its action on three different elements of $\bF$.
\end{proof}

Now it is clear that we have the following consequence of Theorem~\ref{characterization}.

\begin{corollary}\label{GCCauchy}
Let $L\subseteq \bF$ with $n=|L|$ and $2\le k \le n-2$. Let $C$ be a Cauchy code with
scaling map $f:L\rightarrow \F^*$. Then $C$ is a left $G$-code if and only if $G$ is
isomorphic to a subgroup $H$ of order $n$ of $\Gamma_{k,f}$ whose action by homographies
on $L$ is transitive (equivalently, $h(x)\ne x$, for every $x\in L$ and $1\ne h\in H$).
\end{corollary}

As an application of Corollary~\ref{GCCauchy} we characterize when some Cauchy codes with
large location set are group codes.

\begin{lemma}\label{Cambio}
Let $L$ and $L'$ be subsets of $\bF$ with $|L|=|L'|\ge q-2$. Then every Cauchy code with
location set $L$ is equal to a Cauchy code with location set $L'$.
\end{lemma}

\begin{proof}
By assumption $|\bF\setminus L|=|\bF\setminus L'|\le 3$. Since the action of $\Gamma$ on
$\bF$ is triply transitive there is $T\in \GL_2(\F)$ such that $T(\bF\setminus
L)=\bF\setminus L'$ and so $T(L)=L'$. Therefore, by Theorem~\ref{Equal}, if $C$ is a
Cauchy code with scaling map $f:L\rightarrow \F^*$ then $C$ is also a Cauchy code with
scaling map $(\theta_{T,L}^{k-1}\cdot f)\circ T_L\inv:L'\rightarrow \F^*$.
\end{proof}

It is well known that the parity check extension of a narrow sense Reed-Solomon codes is a $q$-ary Cauchy $G$-code of length $q$, for
$G$ the $p$-elementary abelian group of order $q$ (see e.g. \cite{LM}). Next theorem
shows that this is the only possible structure of $G$-code on a $q$-ary Cauchy code of
length $q$.

\begin{theorem}\label{ERS}
Let $C$ be a $q$-ary Cauchy code of length $q$ and dimension $2\le k \le q-2$ and let $G$
be a group of order $q$. Then $C$ is a left $G$-code if and only if it is permutation
equivalent to the parity check extended narrow sense Reed-Solomon code and $G$ is $p$-elementary
abelian.
\end{theorem}

\begin{proof}
One implication was proved by Landrock and Manz \cite{LM}. Conversely, let $C$ be a $q$-ary Cauchy code of length $q$ with
scaling map $f:L\rightarrow \F^*$. By Lemma~\ref{Cambio}, one may assume, without loss of
generality, that $L=\F$. Then $\Gamma_{\F}=\Gamma_{\infty} = N\rtimes H$, with $N$ and
$H$ as in (\ref{NH}). Thus $N$ is a normal Sylow elementary abelian $p$-subgroup of
$\Gamma_{\F}$ and hence it is the only subgroup of order $q$ of $\Gamma_{\F}$. Therefore,
if $C$ is a $G$-code then $G\simeq N$. Moreover, $N$ should be contained in
$\Gamma_{k,f}$ and this implies that $f\circ U_b \equiv f$, for every $b\in \F$, because
$\theta_{U_b}(z)=1$ for every $b\in \F$ and $z\in \bF$. Thus, there is $\lambda\in \F^*$
with $f\circ U_b = \lambda f$. So, $f(0)=f(pb)=(f\circ U_b)((p-1)b) = \lambda f((p-1)b) =
\dots = \lambda^p f(0)$ and hence $\lambda^p=1$. Then $\lambda=1$, since $|\F^*|=q-1$ is
relatively prime with $p$. Hence $f(b)=(f\circ U_b)(0) = f(0)$, for every $b\in \F$, i.e.
$f$ is constant. Therefore $C$ is permutation equivalent to the parity check extended
narrow sense Reed-Solomon code.
\end{proof}

Now we turn our attention to Reed-Solomon like Cauchy codes. It is well known that the
$q$-ary Reed-Solomon codes are cyclic of length $q-1$.

\begin{theorem}\label{RS}
Let $\F=\F_q$, the field of $q$ elements, and let $\xi$ be a primitive element of $\F$.
Let $C$ be a $k$-dimensional $q$-ary Cauchy code with scaling map $f:L\rightarrow \F^*$
and assume that $|L|=q-1$ and $2\leq k\leq q-3$.

Then $C$ is a left group code if and only one of the following conditions hold for some
(for all) $T\in \PGL_2(\F)$ with $T(L)=\F^*$:
\begin{enumerate}
\item there is an integer $m$ such that
    $\left( \theta_{T,L}^{k-1}\cdot f\right ) \circ T\inv_L$ is equivalent to the map
    $f_m:\F^*\rightarrow \F^*$ given by $f_m(z)=z^m$ for every $z\in \F^*$.
\item
$q$ is odd and there are integers $m$ and $m'$ such that $4m+2(k-1)\equiv 2m'+k-1\equiv 0
\text{ mod } (q-1)$ and $\left( \theta_{T,L}^{k-1}\cdot f\right ) \circ T\inv_L$ is
equivalent to the map $f_{m,m'}:\F^*\rightarrow \F^*$ given by
    $$f_{m,m'}(\xi^{2t+r})=\xi^{2tm+rm'}, \quad \text{ for every } t\in \Z
    \text{ and } r=0 \text{ or }1.$$
\end{enumerate}

Moreover, if $G$ is a group of order $q-1$ then $C$ is a left $G$-code if and only if
either $G$ is cyclic and condition (1) holds or $G$ is a dihedral group and condition (2)
holds.
\end{theorem}

\begin{proof}
By Lemma~\ref{Cambio} (and Theorem~\ref{Equal}), it is enough to prove that if $L=\F^*$  and $f(1)=1$ then $C$ is a
left $G$-code if and only if either $G$ is cyclic and $f=f_m$, for some $m$, or $q-1$ is even, $G$
is the dihedral group of order $q-1$ and $f=f_{m,m'}$, for some $m$ and $m'$ satisfying
the conditions of (2).

So assume that $L=\F^*$ and $f(1)=1$. Let $D=\Gamma_{\F^*}=\Gamma_{\{0,\infty\}}
\supseteq \Gamma_{k,f}$.
Then $D=\GEN{A,B}$ with
    $$A=\pmatriz{{cc} \xi & 0 \\ 0 & 1} \quad \text{and} \quad B=\pmatriz{{cc} 0 & 1 \\ 1 & 0}.$$

 So, $D$ is a dihedral group of order
$2(q-1)$. If $q$ is even then the only subgroup of $D$ of order $q-1$ is $\GEN{A}$.
Otherwise, that is if $q$ is odd then, the only subgroups of $D$ of order $q-1$ are
$\GEN{A}$, $\GEN{A^2,B}$ and $\GEN{A^2,AB}$. By Corollary~\ref{GCCauchy}, $C$ is a left
$G$-code if and only if $G$ is isomorphic to a subgroup $H$ of order $q-1$ of
$\Gamma_{k,f}$ which acts transitively on $\F^*$. However $B(1)=1$, and this excludes the
possibility $H=\GEN{A^2,B}$. Thus $C$ is a left $G$-code if and only if either $G$ is
cyclic and $A\in \Gamma_{k,f}$, or $q$ is odd, $G$ is the dihedral group of order $q-1$
and $A^2$ and $M=AB$ belong to $\Gamma_{k,f}$.

Notice that $\theta_A(z)=1$ for every $z\in \F$ and $(f_m\circ A)(z)=f_m(\xi z)=\xi^m
f_m(z)=\xi^m (\theta_A\cdot f_m)(z)$ for every $z\in \F$. Thus $A\in \Gamma_{k,f_m}$ and
hence if $f=f_m$ then $C$ is a cyclic group code. Conversely, if $C$ is a cyclic group
code then $A\in \Gamma_{k,f}$, that is $f\circ A_{\F^*} \equiv f$. Hence there is
$m=0,1,\dots,q-2$ such that $f(\xi z)=\xi^m f(z)$ for every $z$. This implies, together
with the assumption $f(1)=1$, that $f(\xi^i)=\xi^{im}$ and so $f=f_m$.

Secondly assume that $q$ is odd and let $G$ be the dihedral group of order $q-1$. On the one hand, arguing as in the
previous paragraph one deduces that $A^2\in \Gamma_{k,f}$ if and only if there is an integer $0\le m_1 < q-1$ such that
$f(\xi^{2t+r})=\xi^{tm_1}f(\xi)^r$, for every $t\in \Z$ and $r=0$ or $1$. Moreover $m_1$ is even because
$1=f(1)=f\left(\xi^{2\frac{q-1}{2}}\right)=\xi^{\frac{q-1}{2}m_1}$. Set $m_1=2m$. On the other hand, $M(z)=\xi/z$
and $\theta_{M}(z)=z$, for every $z\in \F^*$. Hence $M\in \Gamma_{k,f}$ if and only if there is an integer $0\le m'<
q-1$ such that $f(\xi/z)=\xi^{m'}z^{k-1}f(z)$, for every $z\in \F^*$. So, if $C$ is a left $G$-code then there are
integers $m$ and $m'$ such that
    $$\matriz{{rcl}
    \xi^{m'+(2t+r)(k-1)+2tm} f(\xi)^r &=& \xi^{m'+(2t+r)(k-1)} f(\xi^{2t+r}) =
    f\left(\frac{\xi}{\xi^{2t+r}}\right) = f(\xi^{2(-t)+(1-r)}) \\
    &=& \xi^{-2tm}f(\xi)^{1-r},}$$
for every $t\in \Z$ and every $r=0,1$, or equivalently
    $$\xi^{m'+(2t+r)(k-1)+4tm}=f(\xi)^{1-2r}.$$
Considering the two values $r=0$ and $r=1$ we deduce that $f(\xi)=\xi^{m'}$, hence
$f=f_{m,m'}$, and $2(k-1)+4m\equiv 2m'+k-1 \equiv 0 \mod q-1$. Conversely, if $m$ and
$m'$ satisfy the above conditions and $f=f_{m,m'}$ then, reversing the arguments, we
deduce that $A^2,M\in \PGL_2(\F)_{k,f}$ and so $C$ is a left $G$-code.
\end{proof}

\begin{remarks}{\rm
Let $2\le k \le q-3$ and let $E_m$ and $E_{m,m'}$ denote $k$-dimensional Cauchy codes
with scaling maps $f_m$ and $f_{m,m'}$ respectively. By Theorem~\ref{RS}, every
$k$-dimensional left group Cauchy code of length $q-1$ is permutation equivalent to
either $E_m$, for some $m$, or $E_{m,m'}$, for $m$ and $m'$, satisfying the conditions in
statement (2) of Theorem~\ref{RS}. Using Theorem~\ref{Equal} one can decide when two of
these codes are permutation equivalent:

\begin{enumerate}
\item On the one hand, by Theorem~\ref{Equal}, $E_m$
and $E_{m'}$ are permutation equivalent if and only if $f_{m'}\equiv f_m\circ T_{\F^*}$
for some $T\in \GEN{A,B}$ (see the first part of the proof of Theorem~\ref{RS}). On the
other hand, $f_m\circ A \equiv f_m$ and $f_m\circ B \equiv f_{-m}$. This proves that
$E_m$ and $E_{m'}$ are permutation equivalent if and only if $m'\equiv \pm m \mod (q-1)$.

\item On the other hand, in the second case of Theorem~\ref{RS}, $k$ is necessarily odd and
$f_{m,m'}$ is determined by the values of $2m$ and $m'$ modulo $q-1$. The equation
$2X+k-1\equiv 0 \mod (q-1)$ has a unique solution modulo $\frac{q-1}{2}$, namely
$\frac{1-k}{2}$.

\begin{enumerate}
\item
If $q\equiv -1 \mod 4$ then the equation $4X+2(k-1)\equiv 0 \mod (q-1)$ has also a unique
solution modulo $\frac{q-1}{2}$, namely $\frac{1-k}{2}$ and hence, in this case,
$2m\equiv 1-k \mod q-1$. Thus $f_{m,m'}$ is either $f_{\frac{1-k}{2},\frac{1-k}{2}} =
f_{\frac{1-k}{2}}$ or $f_{\frac{1-k}{2},\frac{q-k}{2}} = f_{\frac{q-k}{2}}$. We conclude
that if $q\equiv -1 \mod 4$ then every left group Cauchy code is a cyclic group code.

\item Otherwise, that is for $q\equiv 1 \mod 4$, the equation $4X+2(k-1) \equiv 0 \mod
(q-1)$ has a unique solution modulo $\frac{q-1}{4}$, namely $\frac{1-k}{2}$. Therefore
there are four possible dihedral group Cauchy codes of this form namely
$E_{\frac{1-k}{2},\frac{1-k}{2}}=E_{\frac{1-k}{2}}$,
$E_{\frac{1-k}{2},\frac{q-k}{2}}=E_{\frac{q-k}{2}}$, $E_{\frac{1+q-2k}{2},\frac{1-k}{2}}$
and $E_{\frac{1+q-2k}{2},\frac{q-k}{2}}$. The last two are not cyclic group codes. To see
this it is enough to show that if $f_{m,m'}\equiv f_{m_1}$ then $2m\equiv 2m_1 \mod
(q-1)$ and $m'\equiv m_1 \mod (q-1)$. This follows by straightforward calculations. This
provides more examples of left group codes which are not abelian group codes.
\end{enumerate}
\end{enumerate}
 }\end{remarks}

\begin{example}{\rm
Let $C=\CC_k(\alpha,1)$ be a Cauchy code over $\F=\F_q$ with location set $L$. We have
seen in Theorems~\ref{ERS} and \ref{RS} that if $L=\F$ or $\F^*$ then $C$ is an abelian
group code. More generally, if $K$ is a subfield of $\F$ and $L=K$ then $C$ is a
$(K,+)$-code and if $L=K^*$ then $C$ is a $(K^*,\cdot)$-code. This can be shown as in the
proofs of Theorems~\ref{ERS} and \ref{RS}.

We claim that the assumption $k\le n-2$ implies that if $C$ is a left group code then
$|L|$ divides $q(q-1)$. Assume that $C$ is a left $G$-code. Then $G$ is isomorphic to a
subgroup of $\Gamma_{k,1}$. If $T=\pmatriz{{cc} a&b\\c&d}$ represents an element of
$\Gamma_{k,f}$, with $c\ne 0$ then there is $\lambda\in \F^*$ such that
$(cz+d)^{k-1}=\theta_{T,L}^{k-1}(z)=\lambda_T$ for every $z\in
L\setminus\{\infty,-\frac{d}{c}\}$. Thus the polynomial $(cz+d)^{k-1}-\lambda_T$ has at
least $n-2\ge k$ roots, a contradiction. This proves that $\Gamma_{k,1}\subseteq
\Gamma_{\infty}$ and hence $|L|=|G|$ divides $|\Gamma_{\infty}|=q(q-1)$, as wanted. This
argument, together with the description of $\Gamma_{\infty}=N\rtimes H$ given in
(\ref{NH}), shows that if $|L|$ is coprime with $q-1$ then $G$ is an elementary abelian
$p$-group and if $|L|$ is coprime with $q$ then $G$ is cyclic.
 }\end{example}

\begin{example}{\rm
Theorems~\ref{ERS} and \ref{RS} describes the left group Cauchy codes over $\F_q$ of
dimension $2\le k \le n-2$ and length $n=q$ or $q-1$. Assume now that $C$ is a Cauchy
code of length $q-2$ and dimension $2\le k \le q-4$. If $f:L\rightarrow \F^*$ is the
scaling map of $C$ then every element of $\Gamma_{k,f}$ permutes the three elements of
$\bF\setminus L$. Since every element of $\Gamma$ is determined by its action on three
elements, we deduce that $\Gamma_{k,f}$ is isomorphic to a subgroup of $S_3$ and hence,
if $L$ is a left $G$-code then $G$ is isomorphic to a subgroup of $S_3$. Therefore
$q-2=|L|=|G|$ divides 6. Since, by assumption $2\le n-2=q-4$, we have $q\ge 6$. We then deduce that $q=8$ and $L$ is a left
$S_3$-code.
 }\end{example}


\begin{thebibliography}{99}

\bibitem{D} A. D\"{u}r, The automorphism groups of Reed-Solomon codes, J. Combin. Theory Ser. A. 44
(1987) 69--82.

\bibitem{H} W.C. Huffman, Codes and groups, in Handbook of coding theory. Vol. II. 1345--1440.
Edited by V. S. Pless, W. C. Huffman and R. A. Brualdi. North-Holland, Amsterdam, 1998.

\bibitem{LM} P. Landrock and O. Manz,
Classical codes as ideals in group algebras. Des. Codes Cryptogr. 2 (1992) 273--285.

\bibitem{SL}
R.E. Sabin and S.J. Lomonaco, Metacyclic Error-Correcting Codes, AAECC, 6 (1995),
191--210.





\end{thebibliography}
\end{document}